\documentclass{sig-alternate-10pt}

\usepackage{subfigure}
\usepackage{url}
\usepackage{paralist}
\def\figurename{Fig.}

\begin{document}
\title{Scale Congestion Control to Ultra-High Speed Ethernet}


\numberofauthors{4}
\author{Wanchun Jiang, Fengyuan Ren, Xin Yue, Chuang Lin}

\maketitle

\thispagestyle{empty}
\makeatletter
\newcommand{\rmnum}[1]{\romannumeral #1}
\newcommand{\Rmnum}[1]{\expandafter\@slowromancap\romannumeral #1@}
\makeatother

\subsection*{Abstract}
Currently, Ethernet is broadly used in LAN, datacenter and enterprise networks, storage networks, high performance computing networks and so on. Along with the popularity of Ethernet comes the requirement of enhancing Ethernet with congestion control. On the other hand, Ethernet speed extends to $40Gbps$ and $100Gbps$ recently, and even $400Gbps$ in the near future. The ultra-high speed requires congestion control algorithms to adapt to the broad changes of bandwidth, and highlights the impacts of small delay by enlarging the bandwidth delay product. The state-of-art standard QCN is heuristically designed for the $1Gbps$ and $10Gbps$ Ethernet, and unaware of the challenges accompanying the ultra-high speed. 

To scale congestion control to ultra-high speed Ethernet, we propose the Adaptive Sliding Mode (ASM) congestion control algorithm, which is simple, stable, has fast and smooth convergence process, can tolerate the impacts of delay and adapt to the wide changes of bandwidth. Real experiments and simulations confirm these good properties and show that ASM outperforms QCN. Designing ASM, we find that the derivative of queue length is helpful to rate adjustment because it reflects the difference between bandwidth and aggregated sending rate. We also argue for enforcing congestion control system staying at the congestion boundary line, along which it automatically slides to stable point. These insights are also valuable to develop other congestion control algorithms in ultra-high speed networks. 

\category{C.2.2}{Computer Systems Organization}{Computer Communication Networks}[Network Protocols] 
\terms{Design, Algorithms, Standardization}
\keywords{Ultra-high Speed Ethernet, Congestion Control, Sliding Mode Motion, Derivative of Queue Length}

\section{Introduction}

Up to now, Ethernet has been very popular in LAN, datacenter and enterprise networks, storage networks, High Performance Computing (HPC) networks, carrier and service provider networks and so on. With the popularity of big date, multiple cores servers and wireless devices, higher speed is required by applications such as huge scientific data transfer, and at positions such as the aggregation and core layer of datacenter and enterprise networks, the aggregation of the wireless access in carrier and service provider network~\cite{roadmap}. The Moore's law also predicts that servers need $100Gbps$ I/O bandwidth in 2017~\cite{commercial}. To solve this problem, the standards of $40Gbps$ and $100Gbps$ Ethernet have been ratified recently~\cite{100G}, and the $400Gbps$ Ethernet is under development~\cite{400G}. Currently, ESnet5 and Internet2 have deployed $100Gbps$ network for scientific experiments~\cite{ESnet5}, companies such as Cisco, Brocade, Extreme and Huawei become suppliers of $100Gbps$ Ethernet~\cite{extreme}, and the $100Gbps$ Ethernet is expected to be available for commercial use in the near future~\cite{commercial}. The ultra-high speed would in turn expand the popularity of Ethernet.

Along with the popularity, however, also come new requirements making congestion management indispensable to ultra-high speed Ethernet. Congestion management is the cornerstone for Ethernet to satisfy the low latency requirement of HPC traffic and the lossless requirement of block-based storage traffic. Moreover, Ethernet congestion management would be helpful to improve the performance of upper layer protocols such as iSCSI~\cite{iSCSI} and iWARP~\cite{iWARP}. In addition, congestion management also takes an important role in enhancing Ethernet as the unified switch fabric of data center networks $\cite{DCB}$. Nowadays, Ethernet congestion management has been standardized by IEEE 802.1Qbb~\cite{802.1Qbb} and IEEE 802.1Qau~\cite{802.1Qau}. Specifically, IEEE 802.1Qbb develops the Priority-based Flow Control (PFC) mechanism, which pause the incoming traffic within the same priority to prevent dropping packets caused by transient congestion. To eliminate the long-lived congestion, IEEE 802.1Qau defines the end-to-end congestion control mechanism composed by a rate-based and queue-based framework, and a standardized congestion control algorithm named Quantized Congestion Notification (QCN). Both PFC and QCN have been supported in commercial devices such as FocalPoint FM6000 series~\cite{FM6000} and Juniper QFX3500~\cite{QFX}.

In this paper, we focus on scaling congestion control algorithm to ultra-high speed Ethernet. The impacts of ultra-high speed on the performance of congestion control are twofold: 1) The available bandwidth can changes in a wide range. Consequently, congestion control algorithms should be elaborately designed to adapt to the broad changes of bandwidth. 2) The ultra-high speed will enlarge the Bandwidth Delay Product (BDP) and correspondingly highlight the impacts of delay even if it's small. Although the good performance of QCN under the $1Gbps$ and $10Gbps$ Ethernet has been shown with simulations, experiments and theoretically analysis ~\cite{sim2,qcn-netfpga, average}, QCN is unaware of these challenges and accordingly unable to scale to ultra-high speed Ethernet (see \S2.3).


We stipulate that the congestion control algorithm for ultra-high speed Ethernet should be simple and stable, converge to stable state quickly  and smoothly, tolerate the impacts of delay and adapt to the broad changes of bandwidth. The Adaptive Sliding Mode (ASM) congestion control algorithm, which is proposed to replace QCN for ultra-high speed Ethernet in this paper, achieves these goals by enforcing congestion control system staying at the congestion boundary line, regardless of Ethernet speed, with the help of the derivative of queue length. Consequently, the congestion control system slides automatically along this boundary line to the stable point, being insensitive to the change of parameters of rate adjustment rules and  delays. In other words, ASM can tolerant the impacts of delay and its convergence process is smooth. Moreover, the sliding process and the approaching process are separated in ASM, and the approaching process can be accelerated independently for short response time. 

We demonstrate a hardware implementation of ASM on the NetFPGA platform~\cite{netfpga}, and use it to verify our simulator on the NS2 platform. Experiments and simulations confirm the theoretical good properties of ASM and show that ASM outperforms QCN, especially with ultra-high speed links.

The key contributions in this paper are not only the ASM algorithm itself, but also the insights on scaling congestion control algorithms to  ultra-high speed networks. We focus on the trajectory of congestion control algorithm on the plane of queue length and aggregated sending rate. We also emphasize the value of the derivative of queue length in determining the direction of trajectory, because it directly reflects the difference between link capacity and aggregated sending rate. Furthermore, we argue for enforcing the trajectory moving straightforwardly to the stable point regardless of link speed, e.g., along the congestion boundary line which has a fixed angle to the efficiency line. In contrast, classic algorithms such as AIMD~\cite{AIMD} and GAIMD~\cite{GAIMD} oscillate up and down around the efficiency line. And adjusting sending rate in proportion to the extent of congestion makes the trajectory cycling around the state point on the plane of queue length and aggregated sending rate.

In the rest paper, we discuss congestion control of ultra-high speed Ethernet in \S2 and present the origin of designing ASM in \S3. We then show the ASM algorithm in detail and theoretically analyze its performance in \S4, and evaluate ASM in \S5. Finally, we discuss related work in \S6 and conclude in \S7.

\section{Congestion Control of Ultra-high Speed Ethernet}
We firstly give a brief introduction about the framework of Ethernet congestion control, then show the challenges of designing congestion control algorithm for ultra-high speed Ethernet, and finally discuss whether QCN can address these challenges or not.

\begin{figure}
\centering
\includegraphics[width=2.8in]{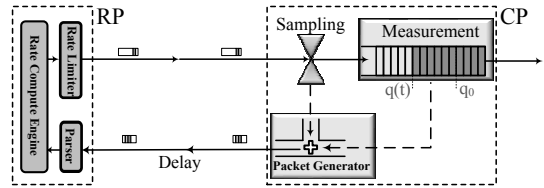}
\caption{Framework defined by IEEE 802.1Qau}
\label{structure}
\end{figure}

\subsection{Framework}
IEEE 802.1Qau has been developing the congestion control mechanism for Ethernet since 2006 $\cite{802.1Qau}$. With the lessons learned from designing other congestion control mechanisms in history, and taking the special environment of link layer into consideration, the researchers developed a rate-based framework, together with several congestion control algorithms, including the final standard QCN. 

Briefly speaking, the framework is composed by two parts, as shown in $\figurename{\ref{structure}}$. One is called \textit{Congestion Point (CP)}, which refers to congested switches. The responsibility of CP is to detect congestion, sample incoming packets, and generate and send feedback packets. At CP, the extent of congestion is represented by information of queue length.  The other one is called \textit{Reaction Point (RP)}, which refers to sources. RP takes charge of parsing feedback packets, and correspondingly adjusting its sending rate according to the feedback state information. In general, the rate adjustment is implemented by the rate limiters at Network Interface Cards (NICs) or edge switches.

Generally speaking, the primary goal of the congestion control algorithms for Ethernet is to tightly hold the instantaneous queue length $q(t)$ around the predefined target point $q_0$. The target queue length $q_0$ is elaborately chosen to avoid overflowing or draining of the buffer, which are corresponding to packets loss and link underutilization, respectively. Ideally, if $q(t)-q_0$ is always small, then the $q_0$ can be set small to reduce the queuing delay and reserve buffer room for burst traffic without degrading link utilization.

\subsection{Challenges}
To vividly illustrate the challenges in ultra-high speed Ethernet, we firstly present a graphic view of congestion control algorithms. Obviously, the core of congestion control algorithm for Ethernet is the interaction between the queue length at CP and the sending rate at RP. Thus, we put the queue length $q(t)$ and the aggregated sending rate $A(t)$ together. At any given time $t$, the state of the congestion control system can be represented by point $(q(t), A(t))$ in the plane of queue length and aggregated sending rate, which is called Q/R plane for convenience of subsequent discussion. Connecting all state points with arrows showing the direction that time $t$ increases, we can obtain the trajectory describing the dynamic behaviors of congestion control system. The red trajectory in $\figurename{\ref{switch-line}}$ is an example and the corresponding dynamic queue variation is illustrated below.

\subsubsection{Congestion Detection}
In the framework defined by IEEE 802.1Qau, the queue length is chosen as the congestion indicator. However, using queue length as the only congestion indicator is one-sided because the extent of congestion is also related with the aggregated sending rate $A(t)$. More intuitively, lines $q(t)=q_0$ and $A(t)=C$ (link capacity) divide the Q/R plane into four regions, but we can only judge that region \Rmnum{1} in $\figurename{\ref{switch-line}}$ is underloaded and region \Rmnum{3} is congested. It is hard to determine whether the rest regions are congested or not.

In history, either line $q(t)=q_0$ is used as the boundary to judge congestion or not in Active Queue Management (AQM) schemes or line $A(t)=C$ is considered as the efficiency line of the AIMD regulation algorithm. In the high speed Ethernet, these simple rules for congestion detection will face challenges since the regions except for \Rmnum{1} and \Rmnum{3} can't be roughly regarded as congested or not. 


\subsubsection{Rate Adjustment}
The direction of trajectory in the Q/R plane has the following attributes. Since 
\begin{equation}
\frac{dq(t)}{dt}=A(t)-C
\label{queue}
\end{equation}
the trajectory of congestion control system can only point to right above the line $A(t)=C$, and point to left when $A(t)<C$. 

Constrained by this law, the rate adjustment rules decide the concrete direction of trajectory. More specifically, the slope $k$ of the tangent of trajectory at any point $(q(t), A(t))$ satisfies 
\begin{equation}
k\triangleq\frac{\Delta A(t)}{\Delta q(t)}\triangleq\lim\limits_{h\rightarrow 0}\frac{A(t+h)-A(t)}{q(t+h)-q(t)}=\frac{\dot{A}(t)}{A(t)-C}
\label{slop}
\end{equation}
as shown in $\figurename{\ref{switch-line}}$. Actually, $\dot{A}(t)$ tells how fast RP adjusts its sending rate.

Eq. (\ref{slop}) implies that rate adjustment rules must associate with bandwidth and aggregated sending rate, or else the slope $k$ varies with the changes of bandwidth. For example, when the additive increase mechanism is employed and the link capacity is huge, $\dot{A}(t)$ is constant in the underloaded region \Rmnum{1}, and accordingly $k\rightarrow0$ referring to Eq. (\ref{slop}). It means that trajectory would move towards line $q(t)=0$, which implies empty buffer and even link underutilization, in the underloaded region \Rmnum{1}. This explain why TCP suffers from serious link underutilization in high speed networks~\cite{HS-TCP}. To maintain desired slop $k$ such that trajectory is able to reach the stable point, rate adjustment rules or $\dot{A}(t)$ should adapt to the changes of bandwidth. As the available bandwidth is unknown and can change in a large range, designing rate adjustment rules for ultra-high speed Ethernet becomes challenging.

\begin{figure}
\centering
\includegraphics[width=2.2in]{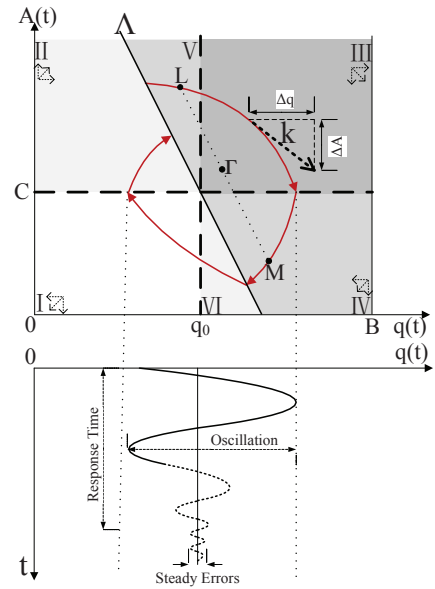}
\caption{Dynamic Behaviors of Congestion Control System in Q/R Plane}
\label{switch-line}
\end{figure}

\subsubsection{Delay}
Contradicted to what has been known that large delay will ruin the stability of congestion control systems, small delay is always ignored. However, ultra-high speed Ethernet will increase the BDP and accordingly highlight the impacts of small delay. For example, when the bandwidth is $100Gbps$ and the delay is $20\mu s$, the BDP is $2Mb$. In other words, the error $E$ of queue length observed by end hosts can be expressed as
\begin{equation}
E=q(t)-q(t-\tau)=\int^t_{t-\tau}{[A(v)-C]}dv\propto \tau*(A-C)
\label{E_1}
\end{equation}
where $\tau$ denote the delay, is in the order of megabits. This is an unavoidable error to all congestion control algorithms for ultra-high speed Ethernet. Consequently, tolerating such a large error is challenging in designing congestion control algorithms.

\textbf{Remark} There are also other errors, whose impacts are magnified by the ultra-high speed similar to delay, such as the measurement error of the queue length in hardware implementation and the absence of feedback information in the sampling interval. Obviously, the amplitudes of these errors are generally smaller than delay. Thus, we just consider delay in this work.

\subsection{Understanding of QCN} 
The state-of-art algorithm QCN is heuristically designed for the $1Gbps$ and $10Gbps$ Ethernet. Although the congestion detection is improved in QCN, it is unaware of the rest challenges and accordingly can't scale to ultra-high speed Ethernet as we will discuss below. For the convenience of discussion, we would list the basic operations in QCN at first and more details are in $\cite{qcn-netfpga}$. 

\subsubsection{Basic Operations}
In QCN, whether the system is congested or not is judged according to the boundary line $F_b(t)=0$ (denoted by $\Lambda$) as shown in $\figurename{\ref{switch-line}}$, where 
\begin{equation}
F_b(t)=-[q(t)-q_0]-w*\Delta Q
\label{feedback}
\end{equation}
and 
\begin{equation}
\Delta Q= \Delta t* \frac{dq(t)}{dt}=\frac{1}{pC}[A(t)-C]
\label{deltaq}
\end{equation}
$\Delta Q$ denotes the variation of the queue length in the sampling interval, $p$ is the sampling probability and $w$ is a weight. QCN employs multiple decrease by utilizing the feedback information and binary search algorithm for rate increase, complemented with correction mechanisms including AI, HAI, Target Rate Reduction and Extra Fast Recovery.

Let $r$ denote the current sending rate and $R$ denote the target sending rate recorded before rate decrease.
\\
\textbf{Rate Decrease (RD)} On the arrival of feedback packet, RP assigns the value of $r$ to $R$ and then decreases $r$. 
\begin{equation}
\left \{ \begin{array}{l}
R\leftarrow r \\
r\leftarrow r(1-G_d |F_b| )
\end{array}
\right.
\label{decrease}
\end{equation}
where $G_d$ is a constant chosen for $G_d |F_{bmax}|=\frac{1}{2}$, i.e., the sending rate won't decrease more than $50\%$ in one adjustment.
\\
\textbf{Fast Recovery (FR)} After the \textbf{RD} phase, RP does binary search. More specifically, RP increases its sending rate $r$ every $T$ seconds 
\begin{equation} 
r\leftarrow \frac{1}{2}(r+R).
\label{fastrecovery}
\end{equation}
Parameter $T$ is acquiescently set to be the time of sending $150KB$ data, which is measured by hardware implemented counter or timer. This behavior will be repeat five times by default, if no feedback messages arrive. 
\\
\textbf{Active Increase (AI)}
After the five times repetition, the binary search is persisted but the target rate $R$ will be increased and the time interval of each cycle is halved. Namely
\begin{equation}
\label{xy} \left \{ \begin{array}{l}
R\leftarrow R+R_{AI} \\
r\leftarrow \frac{1}{2}(r+R)
\end{array}
\right.
\end{equation}
where $R_{AI}$ is a constant unit.
\\
\textbf{Hyper-Active Increase (HAI)}
By default, $R_{AI}$ is set to be $5Mbps$ for the $10Gbps$ Ethernet. But when both the counter and the timer indicate the QCN system has entered into the \textbf{AI} phase, $R_{AI}$ is increased to $50Mbps$. This is called Hyper-Active Increase.

\subsubsection{Congestion Detection}
In the Q/R plane, except for the congested region \Rmnum{3} and the underloaded region \Rmnum{1}, the rest parts are divided into four regions by QCN using lines $\Lambda$, $q(t)=q_0$ and $A(t)=C$, as shown in $\figurename{\ref{switch-line}}$. In region \Rmnum{5}, $A(t)$ is greatly larger than $C$ even though $q(t)<q_0$. And in region \Rmnum{4}, $q(t)$ is greatly larger than $q_0$ even though $A(t)<C$. Therefore, both regions \Rmnum{4} and \Rmnum{5} are treated as congested. Similarly, the regions \Rmnum{2} and \Rmnum{6} are treated as underloaded. Consequently, line $\Lambda$ becomes the congestion boundary line. Note that to solve the challenge of congestion detection, QCN does not explicitly declare whether $A(t)$ or $q(t)$ dominates, but makes trade-off between queue offset and $\Delta Q$ with an adjustable weight $w$ referring to Eq. (\ref{feedback}) and (\ref{deltaq}). However, this good congestion judgment method is not used continually in the rate adjustment of QCN, as we will show subsequently.

\subsubsection{Rate Adjustment}
As the sending rate will be decreased in congested regions, the tangent of the trajectory of QCN can only point to lower right in regions \Rmnum{3} and \Rmnum{5}, and lower left in region \Rmnum{4}, as shown in $\figurename{\ref{switch-line}}$. Constrained by this law, the actual direction of trajectory is decided by concrete rate adjustment rules. Referring to Eq. (\ref{decrease}), there is
\begin{equation}
\dot{A}(t)=G_d|F_b|A(t)\propto |F_b|
\label{decrease-slope}
\end{equation}
in congested regions \Rmnum{3},\Rmnum{4} and \Rmnum{5}. Similarly, the trajectory of QCN can only point to upper right in region \Rmnum{2}, and upper left in region \Rmnum{1} and \Rmnum{6}. In the $i^{th}$ cycle of \textbf{FR} phase, there is
\begin{equation}
\dot{A}(t)=\frac{G_d|F_b|\Sigma_j R_j(t)}{2^iT} \propto |F_b|
\label{increase-slope}
\end{equation}
where $\Sigma_j R_j(t)$ is the aggregated target sending rate, referring to Eq. $(\ref{fastrecovery})$. Totally, the speed of rate adjustment in the \textbf{RD} phase and the \textbf{FR} phase is in proportion to the extent of congestion $|F_b|$ and associated with the aggregated sending rate.

Besides, when the trajectory approaches to line $\Lambda$, there are $\dot{A}\rightarrow0$ referring to Eq. (\ref{decrease-slope}) and (\ref{increase-slope}), and accordingly $k\rightarrow0$ in Eq. (\ref{slop}). Consequently, the trajectory keeps moving from region \Rmnum{5} to region \Rmnum{3}, and from region \Rmnum{6} to region \Rmnum{1}. Referring to the constraints on the direction of trajectory in region \Rmnum{1}, \Rmnum{2}, \Rmnum{3}, \Rmnum{4}, the trajectory of QCN probably cycles around the stable point $(q_0, C)$, just as the red trajectory shown in $\figurename{\ref{switch-line}}$. Correspondingly, the queue length oscillates up and down, and eventually converges to target point $q_0$.

As a step further, the binary search method is initially used to find proper sending rate $C$ in BIC-TCP~\cite{BIC}, where the target sending rate recorded in \textbf{RD} phase is generally considered as larger than the bandwidth. In QCN, the binary search fails to find $C$ in the \textbf{FR} phase when previous \textbf{RD} is done in region \Rmnum{4} such that the target sending rate $\Sigma_j R_j(t)<C$ referring to Eq. (\ref{decrease}). In this condition, QCN needs the \textbf{AI} phase. But the parameters of the  \textbf{AI} phase can not adapt to the broad changes of bandwidth due to the existence of the constant $R_{AI}$. Even for the $10Gbps$ Ethernet, QCN still needs the \textbf{HAI} phase to adjust parameters as aforementioned. When the speed of Ethernet increases to $100Gbps$, where the available bandwidth can change in a huge range, it becomes very hard to set proper parameters for QCN. Previous work $\cite{SMCC}$ also shows that the performance of QCN is sensitive to the changes of both parameters setting and link speed. Therefore, the parameters setting of QCN is cumbersome and can not adapt to the link bandwidth with broad variation. 

The above problem occurs because the rate adjustment goal of QCN to $C$ although whether system is congested or not is judged according to Eq. (\ref{feedback}). Consequently, the speed of rate adjustment of QCN would be the same at three points $L, M, \Gamma$ shown in $\figurename{\ref{switch-line}}$, which are of the same distance to line $\Gamma$ or are judged to be of the same extent of congestion $|F_b|$. But obviously, points $L, M, \Gamma$ denote different extent of congestion and the difference among them is lost when QCN computes the weighted sum according to Eq. $(\ref{feedback})$. While these points can be treat differently to improve the performance of congestion control algorithm as discussed in next section. 

In sum, the rate adjustment of QCN is in proportion to the extent of congestion. But it can not adapt to the ultra-high speed, and the external expression of this issue is that the parameters setting of QCN becomes hard with the increase of bandwidth.

\subsubsection{Delay}
QCN is heuristically designed and the impacts of the delay is not taken into consideration. Recent work $\cite{average}$ shows the upper boundary of delay, which is in the order of hundreds of microseconds when the link capacity is $10Gbps$, for the stable QCN system is large enough in data center networks. Following the work in $\cite{average}$, we deduce a proposition on the lower boundary of delay where QCN becomes unstable (see Appendix A). When the link capacity reaches $100Gbps$, this lower boundary becomes the order of dozens of microseconds. Hence, QCN is easy to becomes unstable with the $100Gbps$ Ethernet due to large delay. Moreover, the larger the bandwidth, the smaller the lower boundary of delay. Therefore, QCN can not scale to ultra-high speed Ethernet due to the impacts of delay.


\section{Origin of Design}
Above challenges place some specialized requirements on Ethernet congestion control. Here we summarize requirements and then show the key ideas steaming from above understanding of QCN to meet these requirements in ASM. These serve as the origin of our design.

An ideal congestion control algorithm for ultra-high speed Ethernet should have the following properties. First, it should be able to converge to the stable point $(q_0, C)$ from any initial state, which is always referred as the stability requirement of the congestion control system. Second, the convergence process should be fast and smooth~\cite{Metrics}, which are measured by the response time and largest amplitude of oscillation, as shown in $\figurename{\ref{switch-line}}$. These are two basic properties of all congestion control algorithms. Third, the congestion control algorithm should be robust to various errors, including delay. In this way, the buffer occupancy can always stay at the target queue length $q_0$ with small oscillation, and accordingly $q_0$ can be set small enough to reduce the queuing delay and reserve spare buffer for burst traffic to reduce the packets loss ratio without degrading link utilization. The last but not the least important, the above properties should not be destroyed by the wide changes of available bandwidth, namely the congestion control algorithm can scale to ultra-high speed Ethernet. With these properties, the congestion control system can admire prefect performance, i.e.,  high throughput, low packets loss ratio and low queuing delay.

Inspired by above understanding of QCN on the Q/R plane, we are aware that $\Delta Q$ is not only useful in congestion detection, as what QCN has done, but also helpful to scale congestion control algorithm to ultra-high speed Ethernet, because $\Delta Q$ reflects bandwidth and  aggregated sending rate referring to Eq. (\ref{deltaq}). More specifically, with the help of $\Delta Q$, we are able to choose rate adjustment rules such that trajectory of congestion control system moves towards any desired direction, i.e., $\Delta A=k*\Delta Q$ or compute the speed of rate adjustment by multiplying the desired tangent $k$ of trajectories by $\Delta Q$, referring to Eq. $(\ref{slop})$.


But what is the desired direction? Intuitively, it is expected that the trajectory moves straightforwardly to the stable point $(q_0, C)$ from any initial points. But this can only be done in regions \Rmnum{2}, \Rmnum{4}, \Rmnum{5} and \Rmnum{6}, because the trajectory is constrained to point to left and right respectively in region \Rmnum{1} and \Rmnum{3}, as shown in \S2.2.2. Taking the physical interpretation of trajectory into consideration, it is natural to expect the trajectory moves along the boundary line $\Lambda$ of congestion  to the stable point. In other words, the congestion control system should be held around line $\Lambda$. Note that once the congestion control system keeps staying at line $\Lambda$, the trajectory automatically slides along it to the stable point because the queue length changes with $A(t)\neq C$ on line $\Lambda$.

As a step further, we find that the above expected behavior is similar to a special motion pattern called sliding mode motion~\cite{Itkis}, whose characteristic is insensitive to the impacts of errors. Consequently, congestion control algorithm designed in this way can address aforementioned challenges, namely adapt to the changes to bandwidth with the help of $\Delta Q$, approach to the stable point quickly and smoothly, and be robust to the impacts of delay in the sliding mode motion. Subsequently, we describe the design of ASM in detail.

\section{Design of ASM}
We present the detailed Adaptive Sliding Mode (ASM) algorithm at first, and then show why it can meet the aforementioned goals by using $\Delta Q$ for its rate adjustment to form the sliding mode motion. Note that we just update the congestion control algorithm but follow the framework defined by IEEE 802.1Qau. 

\subsection{Algorithm}
ASM is also composed by CP and AP. \newline
\textbf{Congestion Point} Similar to QCN, the switch monitors its queue length, ``samples'' incoming packets periodically with probability $p$, and generates feedback packets. The main difference is twofold. First, the offset of the queue length $Q_f\triangleq q(t)-q_0$ and the variation of the queue length $\Delta Q$ are carried in the feedback packets separately. Second, feedback packet is also generated for rate increase. \newline
\textbf{Reaction Point} 
Different from QCN, whose goal is to adjust the aggregated sending rate at $C$, ASM enforces congestion control system staying at boundary line $\Lambda$. Thus, ASM focuses on the direction of trajectories rather than whether the sending rate is increased or decreased, just like in QCN. To control the direction of trajectories, ASM always needs the information of $\Delta Q$ as discussed before. Specifically, let $r(t)$ denote the current sending rate, the concrete rate adjustment algorithm of ASM is as follows.
\begin{equation}
r\leftarrow r-\alpha Q_f-\beta\Delta Q
\label{adjustment}
\end{equation}
where coefficients $\alpha$ and $\beta$ are
\begin{equation}
\alpha=\left \{ \begin{array}{ll}
a^+ &  Q_fF_b> 0  \\
a^- &  Q_fF_b< 0
\end{array}
\right.
\quad
\beta=\left \{ \begin{array}{ll}
b^+ & Q_fF_b> 0  \\
b^- & Q_fF_b< 0
\end{array}
\right.
\label{coefficients}
\end{equation}
Here $a^+, a^-, b^+, b^-$ are constants and $F_b$ is given in Eq. $(\ref{feedback})$. Note that $\Delta Q$ is used for rate adjustment independently in Eq. (\ref{adjustment}), but QCN losses the concrete value of $\Delta Q$ in computing $F_b$ according to Eq. (\ref{feedback}) and uses $F_b$ only for rate decrease. 


\begin{figure}
\centering
\includegraphics[width=2.1in]{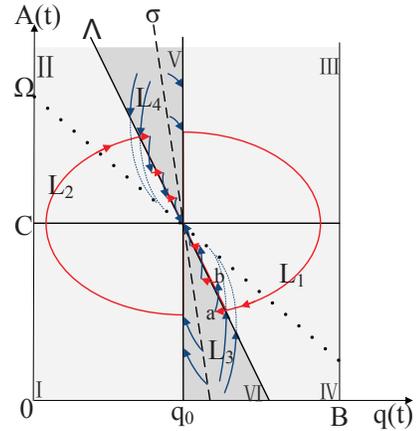}
\caption{Possible Trajectories of ASM}
\label{asm-trajectory}
\end{figure}

\subsection{Trajectory of ASM}
Subsequently, we show that the trajectory of ASM on the Q/R plane will endorse the features described in \S3 if coefficients $a^+, a^-, b^+, b^-$ are properly tuned.

To obtain the trajectory of ASM, we can describe ASM by the following differential equations. (see Appendix B for the detailed deduction)
\begin{equation}
\frac{dx^2(t)}{dt^2}+\frac{N\beta}{pC}\frac{dx(t)}{dt}+\alpha x(t)=0
\label{fluid-flow-new}
\end{equation}
where $x(t)\triangleq q(t)-q_0$ and $N$ is the number of sources. The trajectories of such a second order ordinary differential equation with different characteristic roots have been listed in $\cite{Itkis}$. To enforce the trajectory of ASM staying at the boundary line $\Lambda$, we choose proper coefficients as follows such that trajectory moves towards line $\Lambda$ in both sides of it, as $L_1$ and $L_3$, or $L_2$ and $L_4$ illustrated in $\figurename{\ref{asm-trajectory}}$. Namely
\begin{equation}
\lim\limits_{F_b(t)\rightarrow 0}^{}F_b(t)\frac{dF_b(t)}{dt}\leq 0
\label{smm-condition}
\end{equation}
This inequality is equivalent to 
\begin{equation}
\left \{ \begin{array}{ll}
\frac{w^2N}{p^2C^2}a^--\frac{wN}{p^2C^2}b^-+1< 0 \\
\frac{w^2N}{p^2C^2}a^+-\frac{wN}{p^2C^2}b^++1> 0
\end{array}
\right.
\label{sliding}
\end{equation}
Note that $b^-\neq0$ is sufficient to inequalities (\ref{sliding}), or in other word $\Delta Q$ is indispensable for enforcing the trajectory of ASM at the congestion boundary line. When inequalities (\ref{sliding}) are satisfied, we can know the distribution of the roots of the characteristic equation of $(\ref{fluid-flow-new})$, and accordingly know the trajectory of ASM referring to $\cite{Itkis}$. The detail is summarized in Appendix C.

Specifically, when  $Q_fF_b>0$, i.e., in regions \Rmnum{1},\Rmnum{2},\Rmnum{3} and \Rmnum{4}, the trajectory of ASM is spiral as the red line shown in $\figurename{\ref{asm-trajectory}}$. When $Q_fF_b<0$, i.e., in regions \Rmnum{5} and \Rmnum{6}, the trajectory of ASM is parabola with asymptotic lines $\sigma$ and $\Omega$, as the blue line shown in $\figurename{\ref{asm-trajectory}}$. Consequently, the trajectory of ASM can approach to the boundary line $\Lambda$ quickly from whatever initial state. Once reaching line $\Lambda$, the trajectory will move back and forth between regions \Rmnum{2} and \Rmnum{5}, or between regions \Rmnum{4} and \Rmnum{6}, depending on different feedback information and the corresponding rate adjustment rules $(\ref{adjustment})$, as illustrated in $\figurename{\ref{asm-trajectory}}$. For example, started from region \Rmnum{3}, the trajectory reaches line $\Lambda$ in region \Rmnum{4} spirally. After that, due to various factors such as the delay and the errors discussed in \S2.2.3, the trajectory will pass over line $\Lambda$, go into region \Rmnum{6} at little and stay at point $a$. Subsequently, the other rate adjustment rule with parameters $a^-$ and $b^-$ takes effect in \Rmnum{6} and forces the trajectory moving from point $a$ back to region \Rmnum{4} at point $b$ along a parabola. The switching is repeated uninterruptedly, and the trajectory slides along line $\Lambda$ to the stable point. This motion pattern is called sliding mode motion.

In the sliding mode motion, the trajectory or the dynamic of congestion control system is equivalently described by Eq. $F_b(t)=0$, whose solution is 
\begin{equation}
q(t)=q_0+[q(0)-q_0]e^{-\frac{pC}{w}t}
\label{smm-queue}
\end{equation}
Namely, the queue length $q(t)$ converges to the target point $q_0$ exponentially along line $\Lambda$ in the sliding mode motion. Similarly, the sliding mode motion can also appear between region \Rmnum{2} and region \Rmnum{5}. \\
\textbf{Remark} Note that rate adjustment rules just force the trajectory of congestion control system staying at the boundary line $\Lambda$ of congestion. The trajectory can not stay at any other fixed points but \textit{automatically} slide along line $\Lambda$ to the stable point $(q_0, C)$ because the queue length changes with $A(t)\neq C$ on line $\Lambda$.

\subsection{Convergence and Stability}
From the trajectory of ASM, we can know that ASM can converge to the stable point starting from any initial stat, namely ASM is stable. 

Moreover, the convergence process is composed by two parts: the approaching process refers to ASM moving from any initial point to the boundary line $\Lambda$, and the sliding mode motion where ASM slides along $\Lambda$ to the stable point. Note that the approaching process is directly controlled by the rate adjustment rules, while the sliding mode motion of ASM is equivalently controlled by Eq. $(\ref{smm-queue})$, which has nothing to do with the coefficients of the rate adjustment rules. Consequently, these two parts can be designed separately.

More specifically, the approaching process can be accelerated by setting large coefficients for rate adjustment rules, while small coefficients are required in the sliding mode motion. In general, we suggest the following adaptive parameters setting guidelines.
\begin{itemize}
\item First, all parameters should satisfy inequalities $(\ref{sliding})$. 
\item Large parameters $a^+_A, a^-_A, b^+_A$ and $b^-_A$ are set initially.
\item When $F_b(t)$ is smaller than a bound $B_F$, i.e., when ASM is going to enter into the sliding mode motion, small parameters $a^+_S, a^-_S, b^+_S$ and $b^-_S$ are employed.
\item When ASM is close to the stable point, namely $|Q_f|+|\Delta Q|<B_0$, the large parameters $a^+_A, a^-_A, b^+_A$ and $b^-_A$ are reused in case ASM passively deviates stable point again.
\end{itemize}

With adaptive parameters setting method, the approaching process of ASM is accelerated by large coefficients. And the converging speed of queue length is exponential in sliding mode motion referring to Eq. $(\ref{smm-queue})$. Therefore, the whole convergence speed of ASM will be greatly improved. In addition, since the trajectory of ASM reaches the stable point directly, the amplitude of the oscillations of queue length can be kept small.

\subsection{Impacts of Delay}
The above analysis describes only the ideal behaviors of ASM, where the impacts of delay are ignored. Surprisingly, the behaviors of ASM, even accounting for the impacts of delay, are still a good approximation of the ideal sliding mode motion.

More specifically, we obtain Proposition \ref{Proposition2}, which shows that the impacts of delay can be split into three parts. (see Appendix D). The first one is about the congestion detection, namely the feedback information is delayed. Correspondingly, the slope of the boundary line $\Lambda$ slightly changes, as shown in Eq. (\ref{switching}). The second one is related to the rate adjustment. The computing of sending rate based on old feedback information can be treat as a slight change of the coefficients of rate adjustment rules, and the precise amplitude of the parameters drift is given in Eq. (\ref{para-drift}) and (\ref{amplitude}). The third one is that the computed sending rate takes effect with some lag. This can be treat as error on the result of rate adjustment, which is bounded by inequality (\ref{error}).

Proposition \ref{Proposition2} also explains why the impacts of delay with determined bound can be tolerated by ASM. First, according to $\cite{Itkis}$, the boundary line is enlarged to be a boundary region called quasi-sliding region as illustrated in $\figurename{\ref{smm}}$. Correspondingly, the sliding mode motion is replaced by the quasi-sliding mode motion, which is approximately equivalent to the ideal sliding mode motion. In other words, ASM can still slide to the stable point and the sliding process is equivalently described by Eq. $(\ref{smm-queue})$ in both motion patterns. Second, note that Eq. $(\ref{smm-queue})$ does not involves any of the drifted parameters used in rate adjustment rules $(\ref{adjustment})$. Although the parameters drift would result in a slight change of the direction of trajectory in the approaching process, but the motion pattern of the trajectory persists and the trajectory is always able to reach the boundary region. In total, the impacts of parameters drift can be tolerant by ASM. On the other hand, the drifted parameters can also be compensated through proper parameters setting. Finally, when the impacts of bounded error $E_1$ given in Eq. (\ref{error}) are taken into consideration, ASM can still enter into the sliding mode motion and slide to the stable point outside the region $H_0$ drawn in $\figurename{\ref{smm}}$. With proper parameters setting, the range of $H_0$ becomes small enough according to Eq. (\ref{H0}) and correspondingly the queue length stays close enough to the stable point. Thus, bounded disturbance $E_1$ is tolerable by ASM. 

\begin{figure}
\centering 
\includegraphics[width=2in]{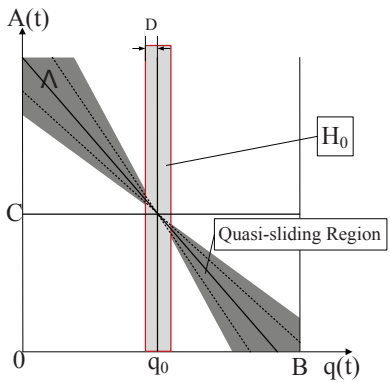}
\caption{Impacts of Delay on ASM}
\label{smm}
\end{figure}

\subsection{Deployment}
Although ASM has the above theoretical advantages, some issues should be considered carefully in deploying ASM. \\
\textbf{CPID}
As multiple bottleneck links exist, the feedback packets of ASM should also involve an identification of the congested switch, called CPID. When the sending rate is decreased, the sender saves the CPID. And the sending rate is increased only when the stored CPID matches the CPID embedded in the feedback packet. In this way, ASM can ignore the feedback information of underloaded switches. Because each packet passes only several switches and one feedback packet is generated when $100$ packets pass one switch (assume $p=0.01$), the number of feedback packets is naturally smaller than that in the case of acknowledging each packet. Thus, the feedback mechanism is light and acceptable.\\
\textbf{Sampling}
As the sampling interval may become smaller than the propagation delay in ultra-high speed Ethernet, a feedback packet may be generated before the former feedback packet arriving the sender. Moreover, when the traffic is burst, these two feedback packets probably go to the same sender. Consequently, one of the feedback packet becomes redundant. In our implementation, the destination address of the feedback packet is recorded and then the incoming packet, whose source address is the same as the recorded destination address, will not be sampled. In this way, two consecutive feedback packets to the same sender can be avoided. \\
\textbf{Weight}
Different rate adjustment rules of ASM work depending on the sign of $Q_f$ and $F_b(t)$. Referring to Eq. $(\ref{feedback})$ and $(\ref{deltaq})$, we can know that the range of the region corresponding to state $Q_fF_b(t)<0$ would be very small in $\figurename{\ref{asm-trajectory}}$. In reality, this region may be skipped over and the sliding mode motion would never occur, because the sampling and the rate adjustment are discrete. Therefore, the weight $w$ should be set relatively large to enlarge the width of the region corresponding to state $Q_fF_b(t)<0$.\\
\textbf{Parameters Setting} The analysis in \S4.2 suggests the parameters setting is guided by inequalities $(\ref{sliding})$. But Eq. $(\ref{para-drift})$ and $(\ref{H0})$ should also be taken into account to reduce both the impacts of parameters drift and the range of region $H_0$. Obviously, all the inequalities suggest large $b^-$ and $a^+$, and small $b^+$ and $a^-$.

In the hardware implementation, the concrete values of parameters $a^+, a^-, b^+$ and $b^-$ are also decided by the number of bits used to represent the queue length, because the amount of the rate changed is also in linear proportion with $Q_f$ and $\Delta Q$, as indicated by Eq. $(\ref{adjustment})$.

\section{Evaluation}
We evaluate ASM using both simulations on the ns2 platform and hardware implementations on the NetFPGA platform $\cite{netfpga}$. The hardware implementation is used to confirm the basic motion pattern of ASM and validates our simulator. As our NetFPGA can only handle traffic at the speed of $1Gbps$ and each card has only four Ethernet ports, simulations dominate the evaluation of ASM.

\subsection{NetFPGA-based Experiments}
The NetFPGA card is developed for fast prototyping of advanced network protocol. We implement the CP of ASM based on the Ethernet Switch project and implement the RP of ASM based on the Packet Generator projects $\cite{netfpga}$. In hardware implementation, both $Q_f$ and $\Delta Q$ are quantized into 8 bits, the buffer size is $128KB$ and the default parameters are that $w=32, B_0=16, B_{\sigma}=64, p=0.01$ and $a^+_A, a^-_A, b^+_A, b^-_A$, $a^+_S$, $a^-_S$, $b^+_S, b^-_S$ are set such that the sending rate changes no more than $\frac{1}{8}, \frac{1}{64}, \frac{1}{16}, \frac{1}{2}, \frac{1}{16}$, $\frac{1}{128}$, $\frac{1}{32},  \frac{1}{4}$ of the link capacity respectively in each rate adjustment.

\begin{figure}
\centering
\subfigure[Queue Length and Rate]{
    \centering%
    \includegraphics[width=1.4in]{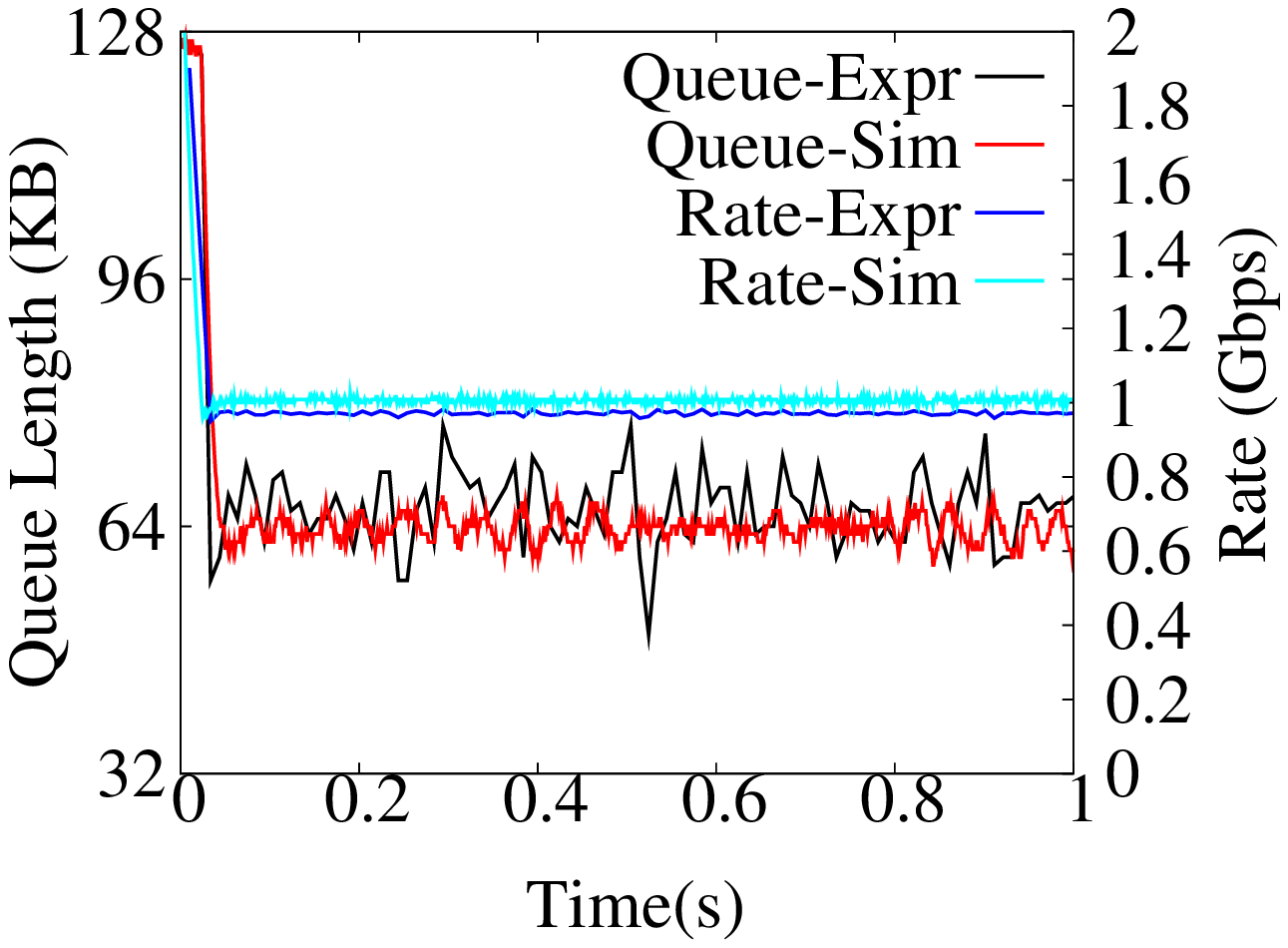}
}
\subfigure[Trajectories]{
    \centering%
    \includegraphics[width=1.4in]{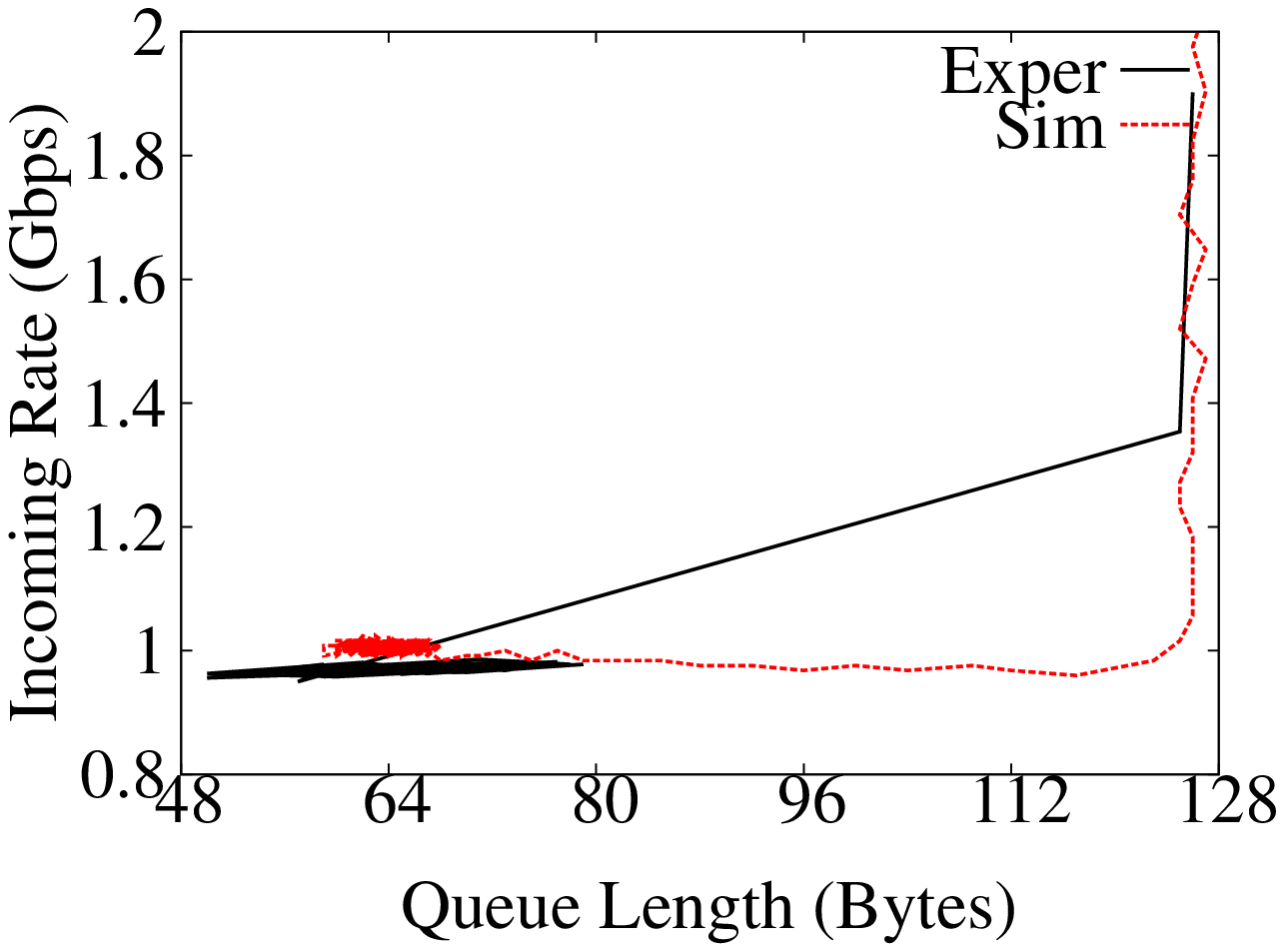}
}
\caption{Sliding Mode Motion of ASM}
\label{sliding-ASM}
\end{figure}

\subsubsection{Sliding Mode Motion of ASM}
Experiment is conducted with 3-source dumbbell topology, $1Gbps$ Ethernet and $2\mu s$ propagation delay at each link. Parameter $q_0=64 packets$ for the convenience of observation. When all the three sources concurrently start long lived flows at the speed of $500Mbps$, the dynamics of queue length and the aggregated sending rate at the congested switch are shown in $\figurename{\ref{sliding-ASM}}$(a). Obviously, both the queue length and the aggregated sending rate converge to the stable point quickly. Putting queue length and aggregated sending rate together, we obtain the trajectory of ASM on the Q/R plane, as shown in $\figurename{\ref{sliding-ASM}}$(b). The trajectory approaches to the stable point directly and finally chatters around the stable point, similar to the ideal trajectory of sliding mode motion. Therefore, the sliding mode motion of ASM is confirmed by experiment. 


\subsubsection{Simulator Validation}
To simulate the hardware implementation, we manually do quantization of the queue length and the sending rate, and add some corresponding random factors. Simulators use the same parameters setting as experiments. The simulation results are also involved in $\figurename{\ref{sliding-ASM}}$.
The differences between experiment result and simulation result are twofold. First, the slopes of the lines, along which ASM approaches to the stable state, are different. This is probably because the interval of measurement, where queue length is read out from NetFPGA to its hosted compute every $10ms$, is comparable to the time taken for ASM to converge to the stable point in experiment. Second, the amplitude of the oscillations of queue length is slightly large in experiment. Such oscillation is expected because the experiment involves many random factors, even though major factors have been considered in the implementation of simulator. Except these differences, the experimental result and the simulation result are identical in $\figurename{\ref{sliding-ASM}}$. Therefore, we believe that the following simulation results are also trustworthy.

\subsection{NS2-based Simulations}
Due to the limited number of ports and link speed, we evaluate ASM with simulation subsequently. QCN is implemented as the comparison algorithm on the ns2 platform by following the pseudo code version 2.3 $\cite{code23}$.
\begin{figure*}
\begin{minipage}{0.33\textwidth}
\centering
\includegraphics[width=2in]{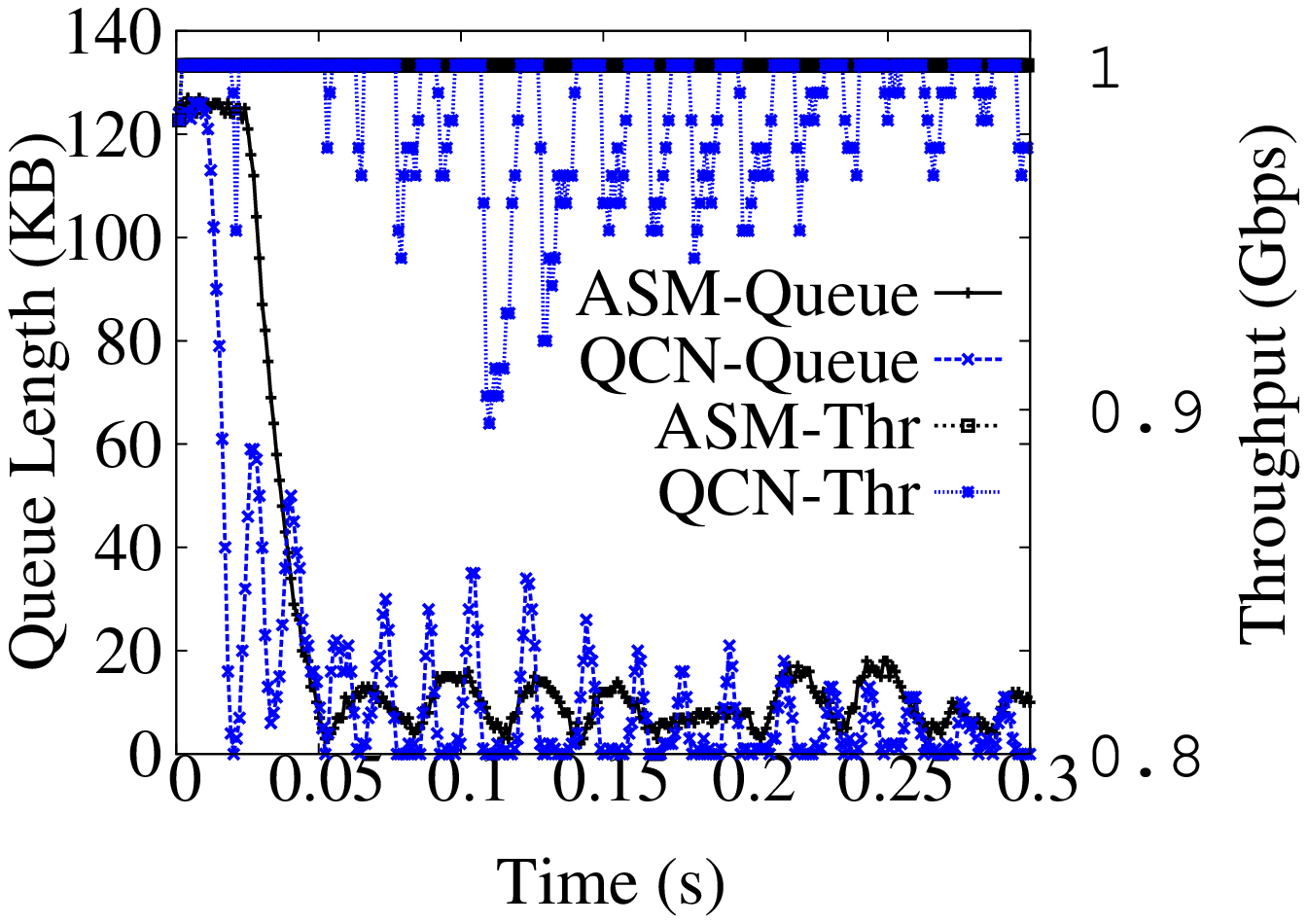}
\caption{Ability of maintaining small queue length}
\label{small-queue}
\end{minipage}
\vspace{0.1cm}
\begin{minipage}{0.33\textwidth}
\centering
\includegraphics[width=2in]{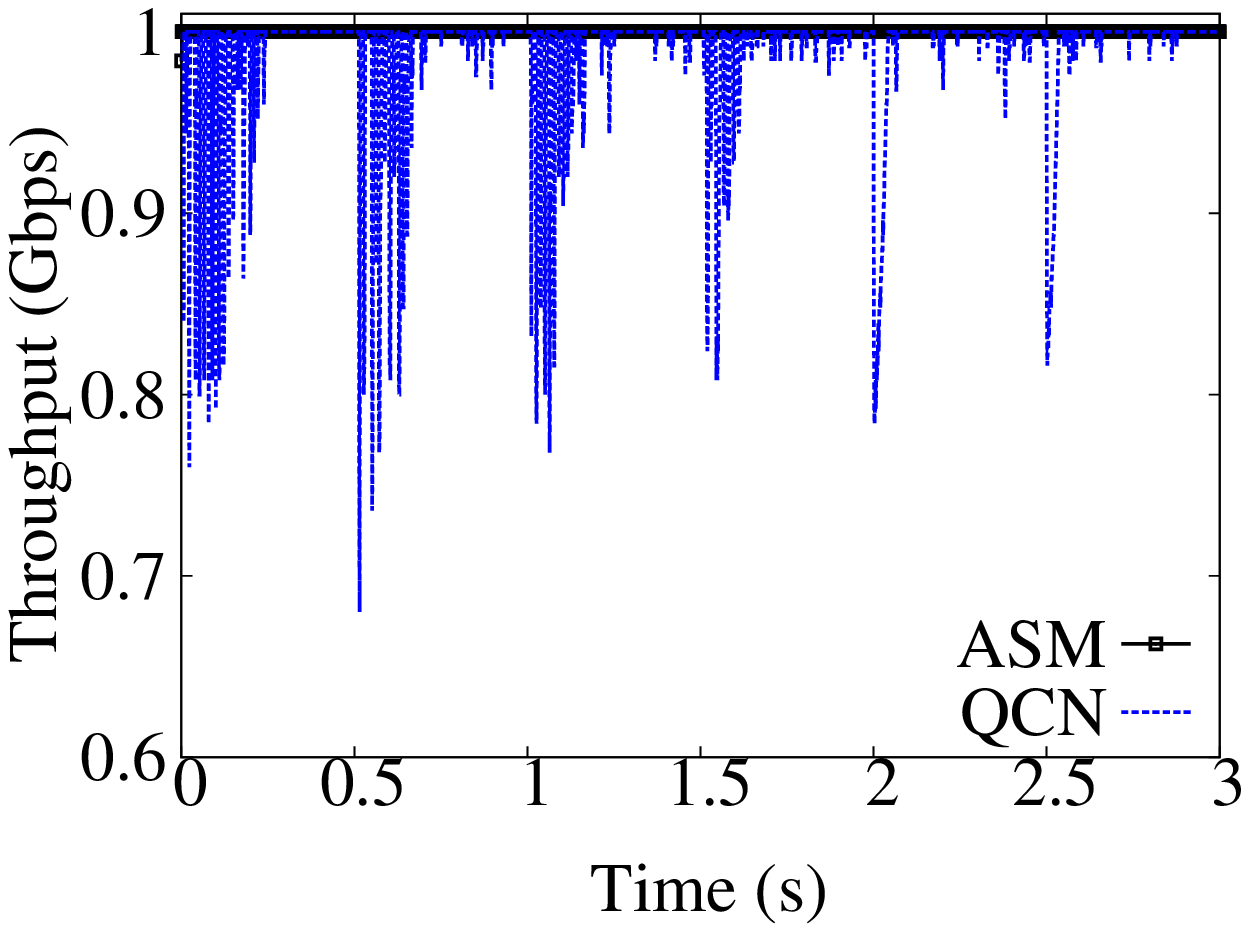}
\caption{Convergence}
\label{convergence}
\end{minipage}
\vspace{0.1cm}
\begin{minipage}{0.33\textwidth}
\centering
\includegraphics[width=1.9in]{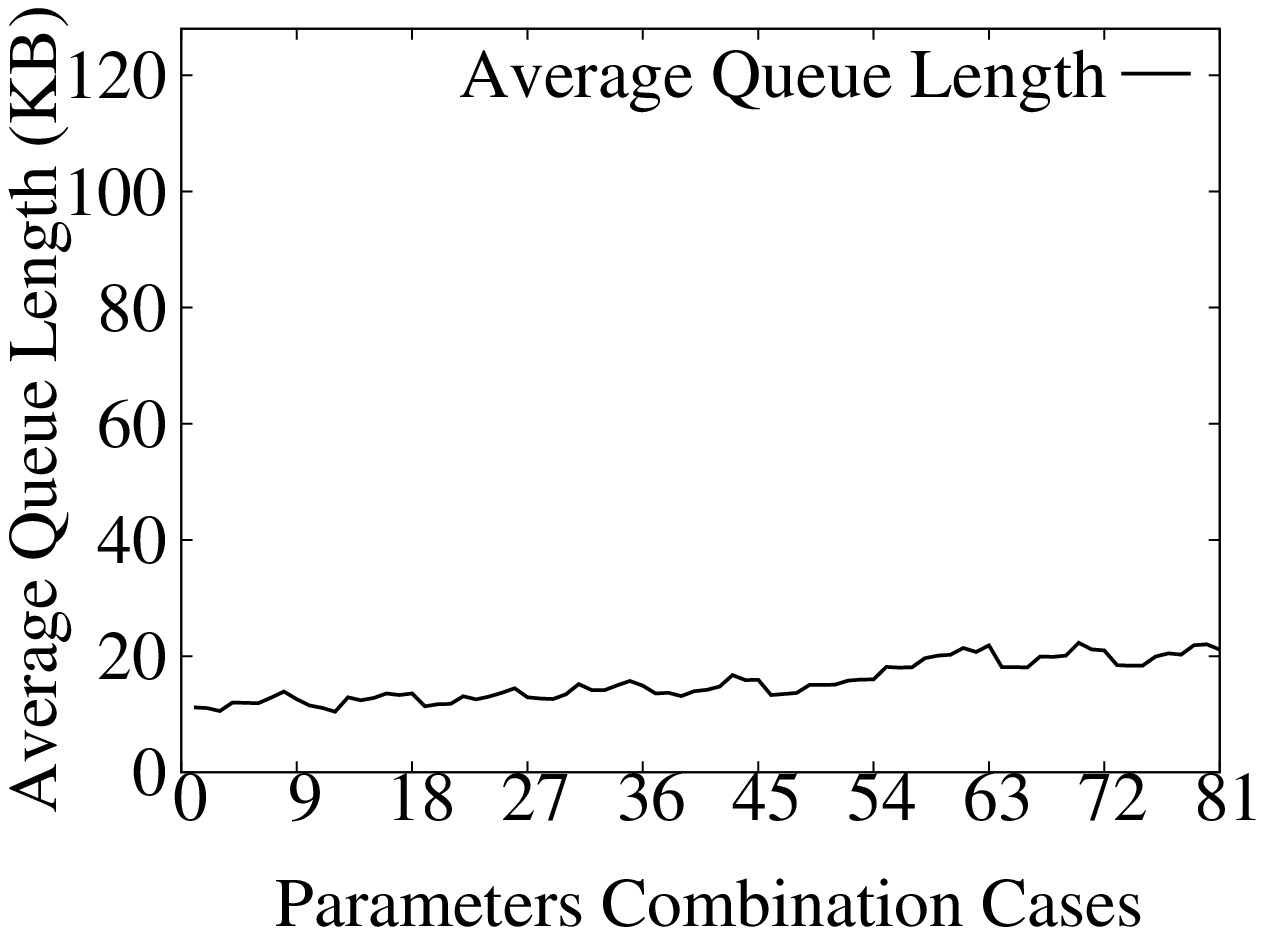}
\caption{Average queue length with different combination of parameters}
\label{parameter}
\end{minipage}
\end{figure*}
\begin{figure*}
\begin{minipage}{0.33\textwidth}
\centering
\includegraphics[width=2in]{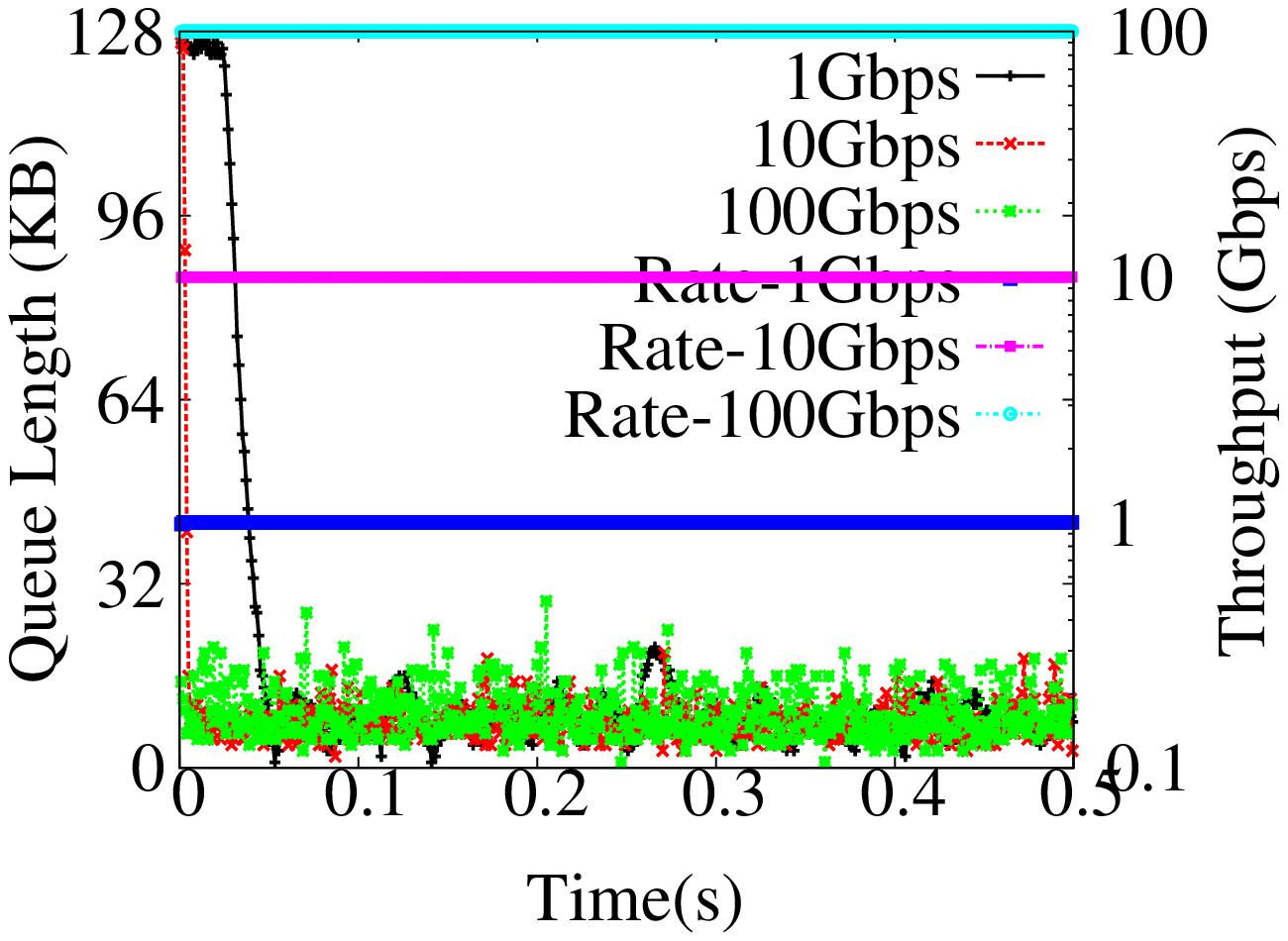}
\caption{Different Bandwidth}
\label{bandwidth}
\end{minipage}
\vspace{0.1cm}
\begin{minipage}{0.33\textwidth}
\centering
\includegraphics[width=2in]{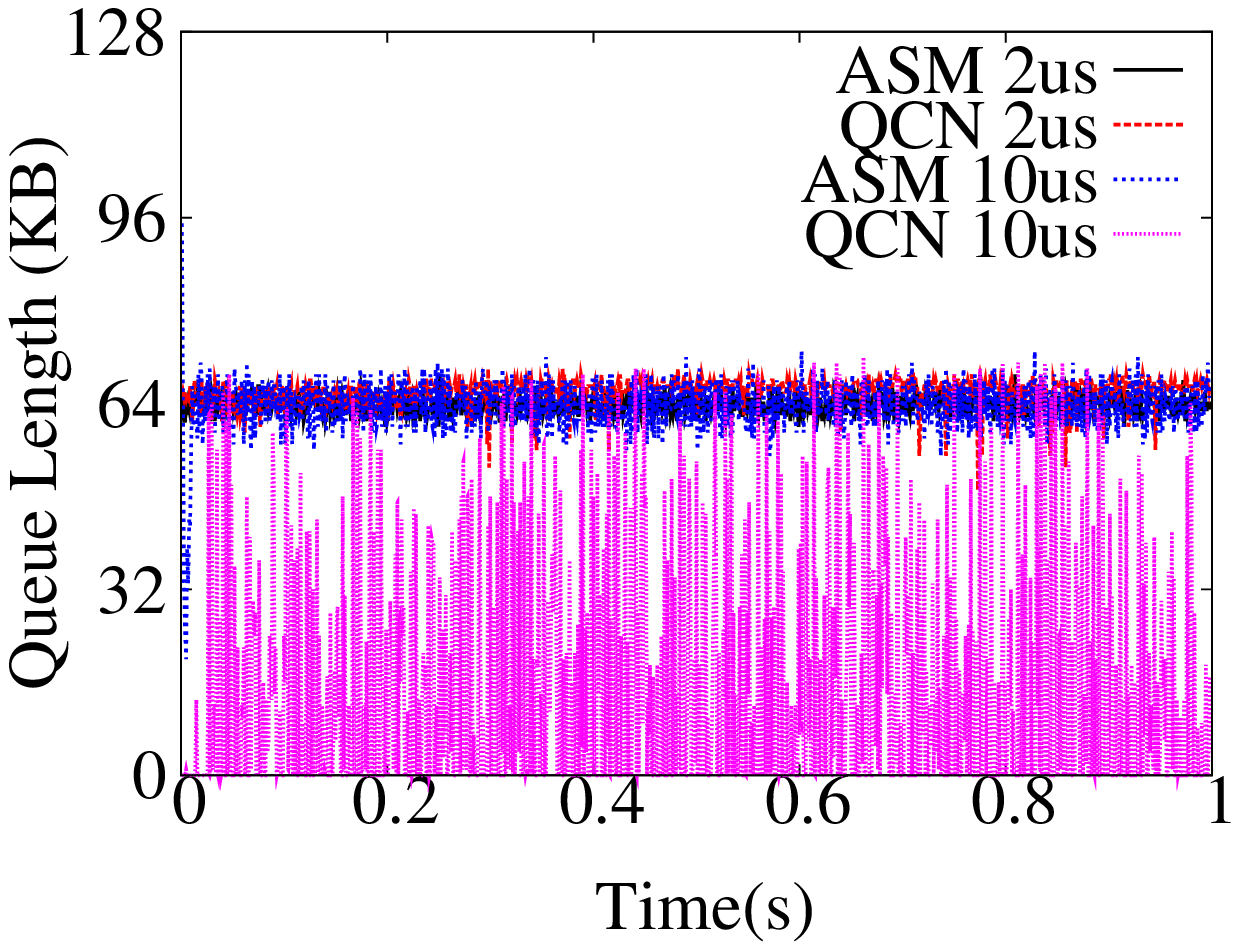}
\caption{Different Delay with $100Gbps$ Ethernet}
\label{delay}
\end{minipage}
\vspace{0.1cm}
\begin{minipage}{0.33\textwidth}
\centering
\includegraphics[width=1.9in]{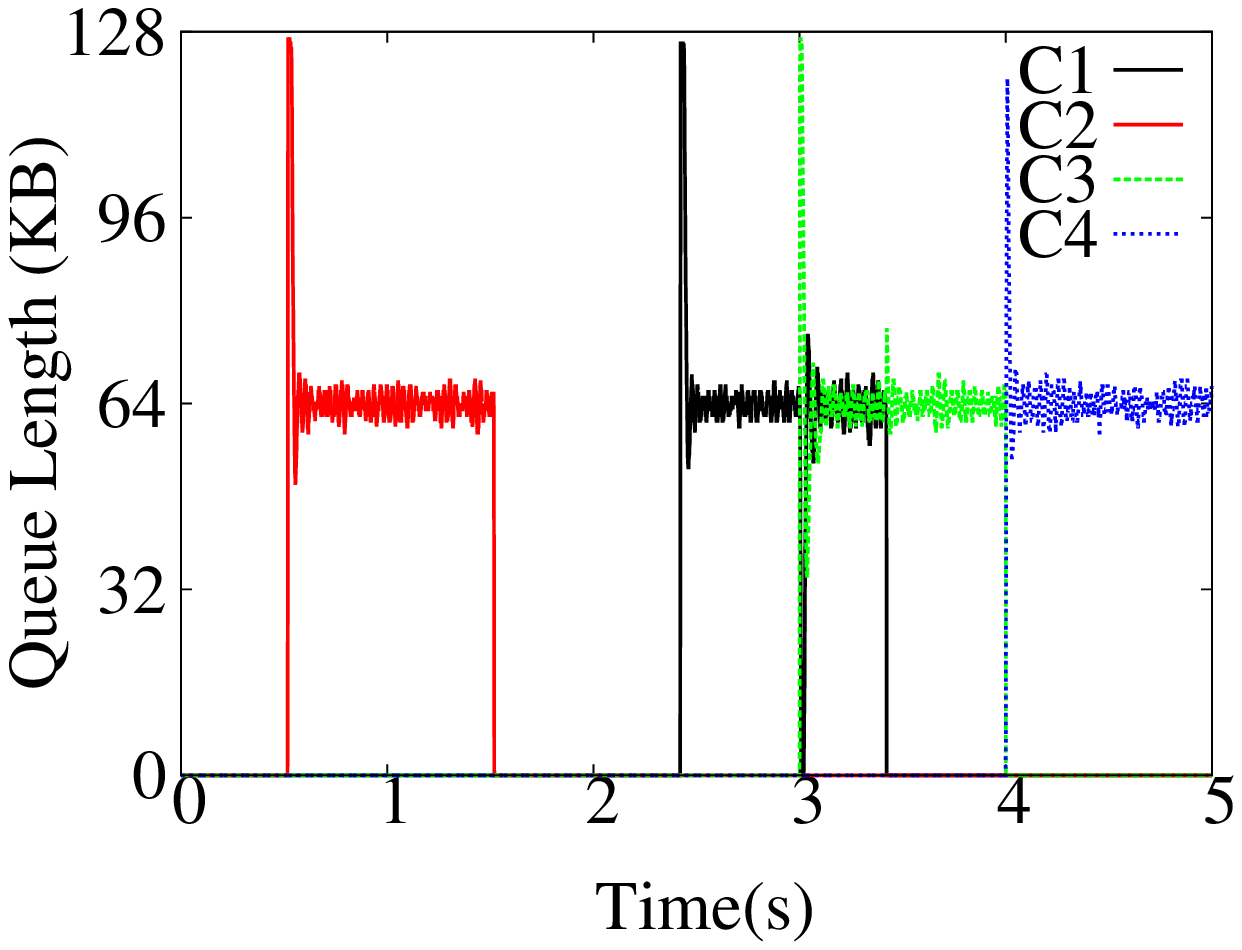}
\caption{Dynamics of queue length at switches with the changes of congested switch}
\label{park-result}
\end{minipage}
\end{figure*}
\textbf{Small Queue Length}
As the direct goals of ASM is to maintain small queue length for low queuing delay, high throughput and so on, we set the target queue length $q_0=5 packets$. The dumbbell topology is still used and the number of sources is increased to $10$. The dynamic of queue length and throughput at the bottleneck link are drawn in $\figurename{\ref{small-queue}}$. Obviously, ASM can always ensure nonempty buffer and does not loss throughput, while QCN frequently drains its buffer and loss throughput.  \\
\textbf{Convergence}
To show the convergence of ASM, we let $5$ sources start at $1Gbps$. Two of them start at the beginning, the rest three sources start subsequently one by one every $0.5s$ latter. At the $2nd$ second, sources are finished one by one every $0.5s$. The instantaneous throughput at the bottleneck link is shown in $\figurename{\ref{convergence}}$. ASM can always maintain high throughput with the incoming or finishing of traffic, while QCN suffers from degrading of throughput. It implies that ASM can converges to the stable point more quickly than QCN with smaller oscillation. \\
\textbf{Parameters} We double or halve parameters $a^+_A, a^-_A$, $b^+_A$, $b^-_A$ respectively, obtaining $3^4=81$ different combinations of these parameters. When each combination is used, we gather the average queue length. As shown in $\figurename{\ref{parameter}}$, the average queue length changes a little even parameters of ASM are changed in a large range. In other words, ASM is insensitive to the changes of parameters. \\
\textbf{Bandwidth}
We modify the link capacity from $1Gbps$ to $100Gbps$, and the dynamic of queue length and the throughput are presented in $\figurename{\ref{bandwidth}}$. With the changes of bandwidth, the amplitude of oscillation almost keeps unchanged, and there is no throughput degradation. Therefore, ASM can adapt to the broad changes of bandwidth.
\\
\textbf{Delay}
We change the propagation delay with $100Gbps$ Ethernet, the dynamics of queue length at the congested switch are shown in $\figurename{\ref{delay}}$. Obviously, both ASM and QCN can maintain small oscillation of the queue length when the delay is extremely small. But when the propagation delay of each link is increased to $10\mu s$, namely when RTT is $60\mu s$, QCN becomes unstable and the buffer of congested switch is emptied frequently. This result agrees with the conclusion in Proposition $\ref{Proposition1}$. While ASM can tolerate the impacts of delay, as claimed in \S4.4.

\noindent\textbf{Multiple Bottlenecks}
Subsequently, we will test the performance of ASM on the classic parking lot topology, where multiple bottlenecks exist. As shown in $\figurename\ref{park-lot}$, all links are of capacity $1Gbps$, flows $F1, F4, F5$ start at $0s, 3s, 4s$ and are terminated at $5s, 4s, 5s$ respectively, and flows $F2, F3$ start randomly at any time in $[0s, 3s]$ and last for one second. Obviously, the congested switches change frequently in this condition. The dynamics of the queue length at switches are shown in $\figurename\ref{park-result}$. According to the simulation results, ASM can always rapidly response to the changes of congested switches.

\begin{figure}
\centering
\includegraphics[width=2.5in]{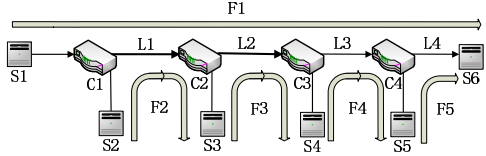}
\caption{Parking Lot Topology}
\label{park-lot}
\end{figure}

\section{Related Work}
As we focus on scaling congestion control algorithm to ultra-high speed in this paper, we mainly compare ASM with the related congestion control algorithms. 

The initial goal of congestion control is throughput. Consequently, traditional TCP keeps on filling the buffer of congested switches until packets are dropped. This method achieves as higher throughput as possible, but suffers from large queuing delay and large packets drop ratio. As low queuing delay and small packets loss ratio are eventually included as goals of congestion control, Active Queue Management (AQM) schemes, such as RED $\cite{red}$ and PI $\cite{PI}$, are developed to help TCP to control the queue length of congested switches at the target points. This target point is always not small to meet the trade-off among throughput, queuing delay and packets loss ratio~\cite{Metrics}. QCN also works in this way. However, smaller queuing delay is desired with the popularity of high speed Ethernet in data center networks, HPC networks and storage networks. Recent work for data center networks such as DCTCP~\cite{DCTCP} and HULL~\cite{HULL} have tried to lower their target queue length as much as possible, even if the throughput is decreased a little. To ensure high throughput, low queuing delay and small packets loss ratio for ultra-high speed Ethernet, we focus on limiting the oscillation of queue length so as to loose contradictions among these goals.

Furthermore, to limit the oscillation of queue length, it is important to ensure congestion control system converging to the stable point quickly and smoothly. But in general, there is a trade-off between response time and smoothness $\cite{Metrics}$. The ultra-high speed make it hard to take this trade-off by greatly shortening the finish time of small flows and enlarge the oscillation. Besides, both response time and smoothness are desired by congestion control algorithms for ultra-high speed Ethernet. Or else, either small flows finish in the slow converge process or packets are dropped due to large oscillation. ASM solves this problem by decoupling response time and amplitude of oscillation to same extent, with the help of sliding mode motion, which is independent to its approaching process.

Traditional rate adjustment algorithms can be roughly divided into two categories. One category computes proper sending rate directly. RCP~\cite{RCP} and XCP~\cite{XCP} are of this category. The key of these algorithms is to collect enough information for precise computing, and the drawback is that the overhead of collecting all information is always high. The other category probes for proper sending rate $C$. The classic AIMD algorithm and its extensions are of this category. The characteristic of these algorithms is that a rate (or window size) decrease rule is used to response to congestion and a rate increase rule is used to probe proper sending rate step by step when these is no congestion. However, the length of each step is hard to be decided due to the bandwidth variation. For example, traditional TCP increases its window constantly, and suffers from poor performance in high speed networks. To adapt to the high speed, GAIMD~\cite{GAIMD}, Scalable TCP~\cite{S-TCP}, High Speed TCP~\cite{HS-TCP}, TCP-Illinois~\cite{illinois} and DCTCP~\cite{DCTCP} changes the factors of AIMD as functions of different congestion indicators such as packet loss ratio, RTT and sequence of ECN bits, and BIC-TCP even replaces addictive increase by binary search. In these ways, the length of step is changed dynamically, generally in proportion to the extend of congestion. However, lacking of information of bandwidth, they always step over the proper sending rate $C$ and accordingly oscillate up and down around the efficiency line. Correspondingly, they cycle around the stable point in the Q/R plane, as discussed in \S2.3.3. In contrast, ASM changes its rate adjustment target from efficiency line $A(t)=C$ to the boundary line $F_b(t)=0$, and approaches to the stable point directly along the boundary line regardless of link speed by utilizing $\Delta Q$, which reflects the difference between link capacity and aggregated sending rate, to control the direction of its trajectory on the Q/R plane.

We also aware that the congestion control system can evolute itself even without any rate adjustment rules. For example, without rate adjustment rules, the queue length increases when $A(t)>C$, and vice versa. This evolution is not emphasized previously but utilized in sliding mode motion. What we have to do is enforce the trajectory of ASM staying at the congestion boundary line with the help of $\Delta Q$ in rate adjustment, and then the trajectory of ASM slides automatically to the stable point because $A(t)\neq C$. 


The sliding mode motion has been used in~\cite{SMVS, PengYan} to analyze the nonlinearity of TCP and design packets marking or dropping rules. But we use it to design rate adjustment algorithm in this paper.

\section{Conclusion}
The ultra-high speed places two challenges on congestion control of Ethernet. One is that the rate adjustment algorithm should adapt to the broad changes of bandwidth. The other is that the impacts of small delay are magnified by the ultra-high speed. The state-of-art algorithm QCN is unaware of these challenges and accordingly can not scale to ultra-high speed Ethernet. 

We stipulate that congestion control of Ethernet should maintain small oscillation of queue length for high throughput, low queuing delay and small packets loss ratio. Utilizing $\Delta Q$ in rate adjustment, ASM can stay at the congestion boundary line and straightforwardly moves along it to the stable point regardless of link speed. Employing sliding mode motion, ASM decouples the response time and the smoothness of convergence process, and accordingly accelerate the convergence process without enlarging the amplitude of oscillations. Moreover, the impacts of the delay on ASM is equal to the combination of the impacts of parameters shift, boundary line drift and a bounded disturbance on the rate adjustment result, all of which are tolerable. Experiments and simulations confirm the good properties of ASM and show that ASM outperforms QCN. 

We emphasize that analyzing congestion control algorithm on the Q/R plane, utilizing the derivative of queue length in rate adjustment and enforcing congestion control system moving along the congestion boundary line to stable point are also helpful to scale other congestion control algorithms to ultra-high speed networks.




\bibliographystyle{abbrv}
\bibliography{sigproc}  



\appendix
\section{Impacts of Delay on QCN}
\newtheorem{theorem}{Proposition}
\begin{theorem}
When the delay $\tau$ satisfies 
\begin{equation}
\tau\geq \frac{1}{\bar{\omega}}[\arctan(\frac{\omega^*}{b})+\arctan(\frac{\omega^*}{\gamma})-\arctan(\frac{\bar{\omega}}{\beta}-\frac{\alpha}{\beta\bar{\omega}})]
\end{equation}
where $\bar{\omega}\triangleq \sqrt{\frac{a_3^2+2\alpha-\beta^2+\sqrt{(a_3^2+2\alpha-\beta^2)^2+4(a_3^2\gamma^2-\alpha^2)}}{2}}$, 
the QCN system is unstable
\label{Proposition1}
\end{theorem}
\begin{proof}
As shown in $\cite{average}$, the characteristic equation of the differential equations describing QCN system is $1+G(s)=0$ where
\begin{equation}
G(s)=e^{\tau s}\frac{a_3(s+b)(s+\gamma)}{s(s^2+\beta s+\alpha)}
\label{translate}
\end{equation}
and $b=p_sR_c$, $\alpha=b(a_1-a_2)$, $\beta=b+a_1$, $\gamma=Na_4/a_3$, $a_1=\frac{\eta(p_s)Rc}{2}+\frac{\eta(p_s)\zeta(p_s)R_{AI}}{2p_s}$, $a_2=\frac{\eta(p_s)R_c}{2}$, $a_3=G_dwR_c$, $a_4=p_sR_cG_d$, $R_c=\frac{C}{N}$, $\eta(p_s)=\frac{p_s}{(1-p_s)^{-100}-1}$ and $\zeta(p_s)=\eta(p_s)(1-p_s)^{500}$. Parameter $C$, $N$, $p_s$ and $\tau$ stand for the link capacity, the number of senders, the sampling probability and the delay, respectively. Define $r(\omega)=|G(j\omega)|$, $\theta(\omega)=-\angle G(j\omega)$, and $\omega^*=\sqrt{a_3^2/2+\sqrt{a_4^3/4+\gamma^2a_3^2}}$, there is $r(\omega^*)<1$. Let $\omega_c$ represents the 0-dB crossover frequency, there is $0<\omega_c<w^*$ and $r(\omega_c)=1$, which is equivalent to 
\begin{equation}
a_3^2(\omega_c^2+b^2)(\omega_c^2+\gamma^2)=\omega_c^2[(\omega_c^2-\alpha)^2+\omega_c^2\beta^2]
\end{equation}
Thus, there is
\begin{equation}
(\omega_c^2-\alpha)^2+\omega_c^2\beta^2>a_3^2(\omega_c^2+\gamma^2)
\label{crossover}
\end{equation}
Solving inequality $(\ref{crossover})$, there is
\begin{equation}
\omega_c\geq \bar{\omega} 
\label{necessary}
\end{equation}
and accordingly 
\begin{equation}
\begin{array}{l}
\theta(\omega_c)=\pi+\omega_c\tau-\arctan(\frac{\omega_c}{b})-\arctan(\frac{\omega_c}{\gamma}) \\
\quad\quad\quad  +\arctan(\frac{\omega_c}{\beta}-\frac{\alpha}{\beta\omega_c}) \\
\qquad\quad\!\!\!\geq  \pi+\arctan(\frac{\omega^*}{b})-\arctan(\frac{\omega_c}{b}) \\
\quad\quad\quad  +\arctan(\frac{\omega^*}{\gamma})+ \arctan(\frac{\omega_c}{\beta}-\frac{\alpha}{\beta\omega_c})\\
\quad\quad\quad   -\arctan(\frac{\omega_c}{\gamma})-\arctan(\frac{\bar{\omega}}{\beta} -\frac{\alpha}{\beta\bar{\omega}}) \\
\qquad\quad\!\!\!\geq \pi
\end{array}
\end{equation}
Therefore, according to the Bode stability criterion~\cite{stability}, QCN is unstable.
\end{proof}
With the recommended parameters in~\cite{average}, the lower boundary of the delay in Eq. $(\ref{necessary})$ is approximately $271\mu s$ for $10Gbps$ Ethernet link, and $27\mu s$ for $100Gbps$ Ethernet link. In other words, delay large than $27\mu s$ would result in unstable QCN system when the speed of Ethernet increases to $100Gbps$. Therefore, the delay should be taken more seriously in designing CN scheme.

\section{Fluid-Flow Model of ASM}
The rate adjustment algorithm $(\ref{adjustment})$ of ASM can be modeled by differential equation
\begin{equation}
\frac{dr(t)}{dt}=-\alpha Q_f-\beta \Delta Q
\label{rate}
\end{equation}
Focus on the bottleneck link and assume that all $N$ sources are identical, i.e., $A(t)=Nr(t)$ and
\begin{equation}
\frac{dq(t)}{dt}=Nr(t)-C
\label{queue-n}
\end{equation}
For the simplification of expression, we make linear transformation
\begin{equation}
\left \{
\begin{array}{l}
x_1(t)\triangleq Q_f=q(t)-q_0 \\
x_2(t)\triangleq \frac{\Delta Q}{pC}=Nr(t)-C
\end{array}
\right.
\label{transform}
\end{equation}
Combining Eq. (\ref{rate}), (\ref{queue-n}), (\ref{transform}) and (\ref{deltaq}), we can obtain the fluid flow model of the ASM system 
\begin{equation}
\left \{
\begin{array}{l}
\frac{dx_1(t)}{dt}=x_2(t)\\
\frac{dx_2(t)}{dt}=-N\alpha x_1(t)-\frac{N\beta}{pC} x_2(t)
\end{array}
\right.
\label{fluid-flow}
\end{equation}
Obviously, Eq. $(\ref{fluid-flow})$ can also be written as the form of Eq. $(\ref{fluid-flow-new})$.

\section{Trajectories of ASM}
Referring to Eq. $(\ref{feedback})$ and (\ref{deltaq}), there are 
\begin{equation}
F_b(t)=-x_1-\frac{w}{pC}x_2
\end{equation}
And referring to Eq. $(\ref{feedback})$ and (\ref{fluid-flow}), there are
\begin{equation}
\begin{array}{ll}
\frac{dF_b(t)}{dt} & =-\frac{dx_1}{dt}-\frac{w}{pC}\frac{dx_2}{dt} \\
                   & =-\frac{dx_1}{dt}+\frac{wN\alpha}{pC}x_1+\frac{wN\beta}{p^2C^2}x_2 \\
                   & =-x_2-\frac{w^2N\alpha}{p^2C^2}x_2+\frac{wN\beta}{p^2C^2}x_2 \\
\end{array}
\label{feedback-dirative}
\end{equation}
when $F_b(t)\rightarrow0$. Substituting Eq. $(\ref{coefficients})$ into Eq. (\ref{feedback-dirative}), we know that $\lim\limits_{F_b(t)\rightarrow 0}^{}F_b(t)\frac{dF_b(t)}{dt}\leq 0$ is equivalent to Eq. $(\ref{sliding})$.

On the other hand, the roots of the characteristic equation of differential equations $(\ref{fluid-flow-new})$ is
\begin{equation}
\lambda_{1,2}=-\frac{N\beta}{2pC}\pm \sqrt{\frac{N^2\beta^2}{4p^2C^2}-\alpha}
\end{equation}
Referring to inequality $(\ref{sliding})$ and~\cite{Itkis}, we can know that $\lambda_{1,2}$ are complex with negative real parts when $Q_fF_b>0$ and the trajectory of ASM is spiral in regions \Rmnum{1}, \Rmnum{2}, \Rmnum{3} and \Rmnum{4}, as the red line shown in $\figurename{\ref{asm-trajectory}}$. When $Q_fF_b<0$, $\lambda_{1,2}$ are two distinct negative real number and the trajectory of ASM is parabola with asymptotic line $\sigma$ and $\Omega$, which denote lines 
\begin{equation}
\sigma=\frac{dx_1}{dt}+\lambda_1x_1=x_2+\lambda_1x_1
\end{equation}
and
\begin{equation}
\Omega=\frac{dx_1}{dt}+\lambda_2x_1=x_2+\lambda_2x_1
\end{equation}
respectively, as the blue line shown in $\figurename{\ref{asm-trajectory}}$.

\section{Impacts of Delay on ASM}
\begin{theorem}
The impacts of the delay are equivalent to the impacts of the combination of a drifted parameters, a drifted boundary line $F_b(t)=0$ and an external disturbance $\bar{E_1}$ on the rate adjustment result, all of which are tolerable by ASM.
\label{Proposition2}
\end{theorem}
\begin{proof}
The proof of this proposition is composed by three parts. The first part is preliminary. The second part is mainly the changes of the fluid flow model of ASM, which deconstructs the impacts of delay. Additional result of the second part is the explanation why the impacts of drifted parameters and drifted boundary line are tolerable by ASM. The third part explain why the external disturbance $\bar{E_1}$ on the rate adjustment result is tolerable in ASM.

\textbf{Part I}
Assume the delay is $\tau$ and define $\epsilon_1=x_1(t-\tau)-x_1(t)$, $\epsilon_2=x_2(t-\tau)-x_2(t)$, $K=\frac{1}{pC}$, $A=N\frac{a^++a^-}{2}$, $K_A=N\frac{a^+-a^-}{2}$, $B=N\frac{b^++b^-}{2}$, $K_B=N\frac{b^+-b^-}{2}$, $K_\Gamma=\sqrt{K_A^2+K_B^2}$, $\Gamma=\sqrt{A^2+B^2}$ and $R(t)=\sqrt{x_1^2(t)+x_2^2(t)}$. There is
\begin{equation}
R(t)\leq |x_1(t)|+|x_2(t)|\triangleq L(t)
\label{Rt}
\end{equation}
and
\begin{equation}
\begin{array}{l}
|\frac{dR(t)}{dt}| 
\leq \frac{1}{R(t)}[R^2(t)+\Gamma R^2(t)+K_\Gamma R^2(t)] \\
\qquad\quad      =(1+\Gamma+K_\Gamma)R(t)
\end{array}
\label{dRt}
\end{equation}
From inequalities $(\ref{Rt})$ and $(\ref{dRt})$, we have
\begin{equation}
\sup\limits_{t\in[t_0,t_0-\tau)}^{}R(t)\leq e^{(1+\Gamma+K_\Gamma)\tau}R(t_0)
\end{equation}
Therefore, the boundary of $\epsilon_1$ can be deduced as follows.
\begin{equation}
\begin{array}{l}
|\epsilon_1(t)| 
=\tau x_2(t-\theta\tau ) \\
\qquad\ \ \leq \tau R(t-\theta\tau ) \\
\qquad\ \ \leq \tau e^{(1+\Gamma+K_\Gamma)\tau }R(t)\\
\qquad\ \ \leq \nu_1L(t)
\end{array}
\label{jwc1}
\end{equation}
where $0<\theta<1$ and $\nu_1\triangleq\tau e^{(1+\Gamma+K_\Gamma)\tau }$. With the same method, we can obtain $|\epsilon_2(t)|\leq \nu_2L(t)$, where $\nu_2\triangleq\tau (\Gamma+K_\Gamma)e^{(1+\Gamma+K_\Gamma)\tau }$.

\textbf{Part II}
The ideal fluid flow model of ASM $(\ref{fluid-flow})$ can be rewritten as
\begin{equation}
\left \{ \begin{array}{l}
    \frac{dx_1(t)}{dt}=x_2(t) \\
    \frac{dx_2(t)}{dt}=-A x_1-KB x_2 \\
\quad\quad\quad\ \  -[K_A|x_1|+KK_B|x_2|]sign\{F_b(t)\}
\end{array}
    \right.
\label{adjust-basic}
\end{equation}
where $sign\{F_b(t)\}$ denotes the sign of $F_b(t)$. When the delay $\tau$ is considered, the boundary line of ASM becomes
\begin{equation}
\bar{F_b}(t)=-[x_1(t)+\frac{w}{pC}x_2(t)+\epsilon_1+w\epsilon_2]=0
\label{switching}
\end{equation}
and the fluid flow model of ASM becomes
\begin{equation}
\left \{
\begin{array}{l}
\frac{dx_1(t)}{dt}=x_2 \\
\frac{dx_2(t)}{dt}=-A(x_1+\epsilon_1)-KB(x_2+\epsilon_2)\\  
\quad\ \ -(K_A|x_1+\epsilon_1|-KK_B |x_2+\epsilon_2|)sign\{\bar{F_b(t)}\}
\end{array}
\right.
\label{adjust-delay}
\end{equation}
Define
\begin{equation}
\left \{
\begin{array}{l}
\varphi=[K_A(|x_1+\epsilon_1|-|x_1|)+ \\
\quad\ \quad  KK_B(|x_2+\epsilon_2|-|x_2|)]sign\{\bar{F_b(t)}\} \\ 
\bar{A}=A+\frac{\varphi(t)}{L(t)}sign\{x_1\} \\
\bar{B}=KB+\frac{\varphi(t)}{L(t)}sign\{x_2\}    
\end{array}
\right.
\label{drift}
\end{equation}
there is
\begin{equation}
\bar{A}x_1+\bar{B}x_2=A x_1+KB x_2+\varphi    
\end{equation}
Accordingly, the fluid flow model $(\ref{adjust-delay})$ becomes
\begin{equation}
\left \{
\begin{array}{l}
\frac{dx_1(t)}{dt}=x_2 \\
\frac{dx_2(t)}{dt}=-\bar{A}x_1-\bar{B}x_2 -A\epsilon_1 -KB\epsilon_2\\
\quad\quad\quad\ \  -[K_A|x_1|+KK_B |x_2|]sign\{\bar{F_b(t)}\}
\end{array}
\right.
\label{algorithm1}
\end{equation}
Define $E_1=A\epsilon_1 +KB\epsilon_2$, there are
\begin{equation}
E_1<(A\mu_1+KB\mu_2)L(t)
\label{error}
\end{equation}
Comparing Eq. (\ref{adjust-basic}) with Eq. (\ref{algorithm1}), we can know that the impacts of the delay is equivalent to the impacts of both disturbance $E_1$ and that parameters $A$, $B$ and $F_b(t)$ are changed to be $\bar{A}$, $\bar{B}$ and $\bar{F_b(t)}$, respectively. As a step further, the change from both $A$ and $B$ to both $\bar{A}$ and $\bar{B}$ is equivalent to that parameters $a^+, b^+, a^-$ and $b^-$ are changed to $\bar{a}^+, \bar{b}^+, \bar{a}^-$ and $\bar{b}^-$ respectively, where
\begin{equation}
\begin{array}{ll}
\bar{a}^+=a^++\frac{\varphi(t)}{L(t)}sign\{x_1\} &
\bar{b}^+=b^++\frac{\varphi(t)}{L(t)}sign\{x_2\} \\
\bar{a}^-=a^-+\frac{\varphi(t)}{L(t)}sign\{x_1\} &
\bar{b}^-=b^-+\frac{\varphi(t)}{L(t)}sign\{x_2\}
\end{array}
\label{para-drift}
\end{equation}
And the largest amplitude of parameters drift is 
\begin{equation}
\frac{|\varphi(t)|}{L(t)}\leq \frac{1}{L(t)}(K_A\epsilon_1+KK_B\epsilon_2)\leq K_A\nu_1+KK_B\nu_2
\label{amplitude}
\end{equation}
According to~\cite{Itkis}, the drifted boundary line changes the sliding mode motion of ASM to be quasi-sliding mode motion, which has the same characteristic as the ideal sliding mode motion. Namely, ASM are still insensitive to the parameters drift of rate adjustment rules and can slide to the stable point in the quasi-sliding region. Moreover, the parameters drift can also be compensated through proper parameters setting.

\textbf{Part III}
When $E_1$ is taken into consideration, the fluid-flow model of the ASM system becomes Eq. $(\ref{algorithm1})$ and the existence of the sliding mode motion is guaranteed by inequality $\lim\limits_{F_b(t)\rightarrow 0}^{}F_b(t)\frac{dF_b(t)}{dt}\leq 0$, which is equivalent to
\begin{equation}
\left \{ \begin{array}{ll}
|x_1|\geq \frac{w^2|E_1|}{|w^2Na^--wNb^-+p^2C^2|} \\
|x_1|\geq \frac{w^2|E_1|}{|w^2Na^+-wNb^++p^2C^2|}
\end{array}
\right.
\label{sufficient-sliding}
\end{equation}
Define region $H_0$ as $\{q(t):|q(t)-q_0|\leq D\}$, where the width $D$ is  
\begin{equation}
\begin{array}{l}
D\triangleq\max\{\frac{w^2|E_1|}{|w^2Na^--wNb^-+p^2C^2|}, \\
\quad\quad \frac{w^2|E_1|}{|w^2Na^+-wNb^++p^2C^2|} \}
\end{array}
\label{H0}
\end{equation}
\end{proof}

\end{document}